\documentclass[twocolumn,preprintnumbers,amsmath,amssymb]{revtex4}
\usepackage{amsthm,graphicx,multirow,bm,diagbox,color,graphicx}% Include figure files
\usepackage{dcolumn}% Align table columns on decimal point
\usepackage{bm}% bold math
\usepackage[]{hyperref}

\newtheorem{theorem}{Theorem}
\newtheorem{corollary}[theorem]{Corollary}

\newtheorem{example}{Example}
\usepackage{titlesec}
\usepackage{tikz,float}
\usetikzlibrary{shapes.geometric, arrows}
\makeatletter
\renewcommand\thetable{\@arabic\c@table}
\makeatother

\begin{document}

%\preprint{APS/123-QED}

\title{Purity and construction of arbitrary dimensional $k$-uniform mixed states}% Force line breaks with \\

\author{Xiao Zhang$^1$$^2$$^3$, Shanqi Pang$^3$\footnote{Email: shanqipang@126.com}, Shao-Ming Fei$^4$, Zhu-Jun Zheng$^5$}

\affiliation{$^1$Post-doctoral innovation practice base, \\
Henan Institute of Science and Technology, Xinxiang, 453003, China\\
$^2$Postdoctoral Research Center of Henan Bainong Seed Industry Co., Ltd,\\
Xinxiang, 453003, China\\
$^3$College of Mathematics and Information Science, Henan Normal University,\\
 Xinxiang, 453007, China\\
$^4$School of Mathematical Science, Capital Normal University, Beijing, 100048, China\\
$^5$School of Mathematics, South China University of Technology, Guangzhou, 510641, China}
%\date{\today}

\begin{abstract}
$k$-uniform mixed states are a significant class of states characterized by all $k$-party reduced states being maximally mixed. Novel methodologies are constructed for constructing $k$-uniform mixed states with the highest possible purity. By using the orthogonal partition of orthogonal arrays, a series of new $k$-uniform mixed states is derived. Consequently, an infinite number of higher-dimensional $k$-uniform mixed states, including those with highest purity, can be generated.

{\bf Key words: } $k$-uniform mixed states; highest purity; orthogonal arrays; orthogonal partition
\end{abstract}

\maketitle

\section{Introduction}

Quantum entanglement is crucial in quantum information theory. Entangled states enhance various quantum information processing tasks such as teleportation \cite{bou97,rie04}, quantum key distribution \cite{lo12}, dense coding, error-correcting codes \cite{arn13,sco04}, and quantum computation \cite{jozsa03}. However, the characterization of multipartite entanglement remains a significant unresolved issue in quantum entanglement theory \cite{goy16}.

A noteworthy class of pure entangled states, known as $k$-uniform states, has recently attracted considerable interest. An $N$-partite pure state is $k$ uniform if all its $k$-partite reduced density matrices are maximally mixed \cite{sco04,goy14}, and such states are also termed $k$-multipartite maximally entangled pure states \cite{arn13}. Particularly, when $k=\lfloor \frac{N}{2} \rfloor$ ($\lfloor\rfloor$ is the floor function), the states are referred to as absolutely maximally entangled (AME) states, denoted as AME$(N,d)$ \cite{hel12}, where $d$ represents the dimension of the individual systems. Methods such as graph states \cite{helwig}, orthogonal arrays (OAs) \cite{goy14,goy16,npj,qinp,jpa,limao19,shifei22}, mutually orthogonal Latin squares and cubes \cite{goy15}, symmetric matrices, and algebraic geometry codes \cite{feng17} have been employed to construct maximally entangled pure states for both homogeneous and heterogeneous systems.

Nevertheless, $k$-uniform pure states such as AME$(4,2)$ and AME$(7,2)$ \cite{huber17} do not exist. K{\l}obus et al. \cite{kmix} posed a question: if $k$-uniform pure states cannot be constructed, what is the highest possible purity for a $k$-uniform mixed state for a given number of parties $N$? They introduced a method for explicitly constructing $k$-uniform states using $N$-qudit operators and provided several examples for $k$-uniform qudit mixed states. Nevertheless, constructing $k$-uniform mixed states, particularly for high-dimensional systems, remains challenging.

This study explores the purity and construction of arbitrary dimensional $k$-uniform mixed states. In Section \ref{zh}, we convert the problem into constructing a matrix that meets specific criteria and present new $k$-uniform mixed states, including those with the highest purity. Section \ref{gw} details a general method for constructing $k$-uniform mixed states using orthogonal partitions of OAs, enabling the creation of infinitely many higher-dimensional $k$-uniform mixed states with maximal purity. The purity of the obtained states is discussed in Section \ref{dis}. We present the differences between our findings and those reported in Ref. \cite{kmix} and conclude in Section \ref{con}.

\section{$k$-Uniform Mixed Qubit States from Matrices over Galois Field GF(4)}\label{zh}

K{\l}obus et al. \cite{kmix} presented a scheme for constructing $k$-uniform states from particular sets of $N$-qubit Pauli matrices. Let $\sigma_x$, $\sigma_y$ and $\sigma_z$ be the standard Pauli matrices, and $\mathbb{I}_n$ be the identity operator of order $n$. Suppose that a set of $N$-qubit Pauli operators $\mathcal{G}=\{G_1,G_2,\ldots,G_m\}$ exists, where $G_i=U_1\otimes\cdots\otimes U_N$ with $U_i\in\{I,\sigma_x,\sigma_y,\sigma_z\}$, such that\\
(i) mutual commutative, $[G_i,G_j]=0$ for all $i$, $j$;\\
(ii) independence, $G_1^{i_1}G_2^{i_2}\cdots G_m^{i_m}\sim \mathbb{I}_{2^N}$ only when $i_1=\ldots=i_m=0$ with $i_j=\{0,1\}$, namely, $G_1^{i_1}G_2^{i_2}\cdots G_m^{i_m}$ is proportional to the identity operator $\mathbb{I}_{2^N}$ only when $i_1=\ldots=i_m=0$ with $i_j=\{0,1\}$;\\
(iii) $k$ uniformity, the $N$-qubit Pauli operator $G_1^{i_1}G_2^{i_2}\cdots G_m^{i_m}$ $(i_j=\{0,1\})$ contains the identity operators on at most $N-k-1$ positions. The elements of $\mathcal{G}$ are called generators. A $k$-uniform state is generated by the elements of $\mathcal{G}$ as follows:
$$\rho=\frac{1}{2^N}\sum_{j_1,\ldots,j_m=0}^{1}G_1^{j_1}\cdots G_m^{j_m}=\frac{1}{2^N}(\mathbb{I}+G_1)\cdots(\mathbb{I}+G_m).$$
For the case of $m=N$, this construction leads to a $k$-uniform pure state with purity ${\rm Tr}(\rho^2)=1$. Otherwise, $\rho$ is a mixed state.

{\sf Remark:} The independence of $\{G_1,G_2,\ldots,G_m\}$ is analogous to the linear independence of a set of vectors $\{v_1, v_2, \ldots, v_m\}$, which is said to be linearly independent if and only if the equation $c_1 v_1 + c_2 v_2 + \cdots + c_n v_n = 0$ has the unique solution $c_1 = c_2 = \cdots = c_n = 0$. Similarly, a set of $\{G_1,G_2,\ldots,G_m\}$ is said to be independent if and only if the equation $G_1^{i_1}G_2^{i_2}\cdots G_m^{i_m}\sim \mathbb{I}_{2^N}$ has the unique solution $i_1=\ldots=i_m=0$.

Accordingly, the state $\rho$ generated by $\mathcal{G}$ has a pure state decomposition, $\rho=\sum_ip_i|\psi_i\rangle\langle\psi_i|$, given by
$2^{N-m}$ pure states $|\psi_1\rangle$, $|\psi_2\rangle$, $\ldots$, $|\psi_{2^{N-m}}\rangle$ and $p_i=\frac{1}{2^{N-m}}$. Furthermore, all the $k$-partite reduced density matrices of $\rho$ are maximally mixed, in other words, all the $k$-partite reduced density matrices are proportional to the identity matrix. Therefore, $\rho$ is a $k$-uniform state.

$k$-uniform states can also be obtained by constructing particular matrices. For simplicity, denote the elements $\{0,1,x,x+1\}$ of the Galois field GF(4) by $0, 1, 2, 3$, respectively. The mapping $\sigma_0 \rightarrow0$, $\sigma_1\rightarrow1$, $\sigma_2\rightarrow2$ and $\sigma_3\rightarrow3$ leads to the correspondence between the multiplication of the set $\{\sigma_0,\sigma_1,\sigma_2,\sigma_3\}$ and the addition of the set $\{0,1,2,3\}$ in GF(4). Then an $N$-qubit Pauli operator $G_i=\underbrace{\sigma_j\otimes\cdots\otimes\sigma_k}_N$ can be written as a $1\times N$ matrix $(\underbrace{j\cdots k}_N)$, and the set $\mathcal{G}=\{G_1,G_2,\ldots,G_m\}$ provides a $m\times N$ matrix. For example, the $N$-qubit Pauli operator $\sigma_1\otimes\sigma_2\otimes\sigma_3\otimes\sigma_0\otimes\sigma_2$ can be written as a $1\times 5$ matrix $(1 2 3 0 2)$ and $\sigma_2\otimes\sigma_0\otimes\sigma_3\otimes\sigma_1\otimes\sigma_3$ can be written as a $1\times 5$ matrix $(2 0 3 1 3)$. Therefore, the product $G_1G_2=(\sigma_1\otimes\sigma_2\otimes\sigma_3\otimes\sigma_0\otimes\sigma_2)
(\sigma_2\otimes\sigma_0\otimes\sigma_3\otimes\sigma_1\otimes\sigma_3)
=\sigma_1\sigma_2\otimes\sigma_2\sigma_0\otimes\sigma_3\sigma_3\otimes\sigma_0\sigma_1\otimes\sigma_2\sigma_3
=(i\sigma_3)\otimes\sigma_2\otimes\sigma_0\otimes\sigma_1\otimes(i\sigma_1)
=-\sigma_3\otimes\sigma_2\otimes\sigma_0\otimes\sigma_1\otimes\sigma_1$ corresponds to the sum of $(1 2 3 0 2)$ and $(2 0 3 1 3)$, i.e., $(3 2 0 1 1)$.

From a set of $N$-qubit Pauli operators $\mathcal{G}=\{G_1,G_2,\ldots,G_m\}$, we obtain an $m\times N$ matrix $G$ with elements 0, 1, 2, 3. Let $\mathbb{Z}_d^m$ be a $d^m\times m$ matrix whose rows are all possible $m$-tuples over a ring $\mathbb{Z}_d=\{0, 1, \ldots, d-1\}$. Considering the  conditions (i)-(iii) above, we arrive at the following conclusion.

\begin{theorem}\label{gsan}
Let $G_{m\times N}$ be a matrix with elements of GF(4) such that\\
(a) for any two rows $a_i=(a_{i1},a_{i2},\ldots,a_{iN})$ and $a_j=(a_{j1},a_{j2},\ldots,a_{jN})$ of $G$, $N-|\{k(1\leq k\leq N):a_{ik}=0 \ or \ a_{jk}=0 \ or \ a_{ik}=a_{jk}\}|$ is even,\\
(b) the matrix $\mathbb{Z}_2^mG$ has only one zero row,\\
(c) except for the one zero row of matrix $\mathbb{Z}_2^mG$, at most $N-k-1$ zeros are present in each of the remaining rows.\\
Therefore, the purity of $N$-qubit $k$-uniform state given by $G_{m\times N}$ is $2^{m-N}$.
\end{theorem}

\begin{proof}
Let two rows $a_i=(a_{i1},a_{i2},\ldots,a_{iN})$ and $a_j=(a_{j1},a_{j2},\ldots,a_{jN})$ correspond to the $N$-qubit Pauli operators $G_i=\sigma_{a_{i1}}\otimes\sigma_{a_{i2}}\otimes\cdots\otimes\sigma_{a_{iN}}$ and $G_j=\sigma_{a_{j1}}\otimes\sigma_{a_{j2}}\otimes\cdots\otimes\sigma_{a_{jN}}$, respectively. Accordingly, we have\\
\centerline{$G_iG_j=\sigma_{a_{i1}}\sigma_{a_{j1}}\otimes\sigma_{a_{i2}}\sigma_{a_{j2}}\otimes\cdots\otimes\sigma_{a_{iN}}\sigma_{a_{jN}}$} and\\ \centerline{$G_jG_i=\sigma_{a_{j1}}\sigma_{a_{i1}}\otimes\sigma_{a_{j2}}\sigma_{a_{i2}}\otimes\cdots\otimes\sigma_{a_{jN}}\sigma_{a_{iN}}.$} If $a_{ik}=0$ or $a_{jk}=0$, then $\sigma_{a_{ik}}\sigma_{a_{jk}}=\sigma_{a_{jk}}\sigma_{a_{ik}}$. If $a_{ik}=a_{jk}$, then $\sigma_{a_{ik}}\sigma_{a_{jk}}=\sigma_{a_{jk}}\sigma_{a_{ik}}=\sigma_0$. Therefore, $\sigma_{a_{ik}}\sigma_{a_{jk}}=\sigma_{a_{jk}}\sigma_{a_{ik}}$ only if $a_{ik}=0 \ or \ a_{jk}=0 \ or \ a_{ik}=a_{jk}$. Otherwise, $\sigma_{a_{ik}} \sigma_{a_{jk}}=-\sigma_{a_{jk}}\sigma_{a_{ik}}$. Since $N-|\{k(1\leq k\leq N):a_{ik}=0 \ or \ a_{jk}=0 \ or \ a_{ik}=a_{jk}\}|$ is even, $G_iG_j=G_jG_i$. Therefore, $[G_i,G_j]=0$ for all $i$ and $j$.

It is evident that $G_1^{i_1}\cdots G_m^{i_m}$ ($i_j=\{0,1\}$) corresponds to the rows of the matrix $\mathbb{Z}_2^mG$. Accordingly, $G_1^{i_1}\cdots G_m^{i_m}\sim \mathbb{I}_{2^N}$ for $i_1=\ldots=i_m=0$ if and only if the matrix $\mathbb{Z}_2^mG$ has only one zero row. $G_1^{i_1}\cdots G_m^{i_m}$ contains at most $N-k-1$ identity operators if and only if the matrix $\mathbb{Z}_2^mG$ has  at most $N-k-1$ zeros in each row, except for the one zero row.
\end{proof}

By using Theorem \ref{gsan}, we next present matrices that satisfy the conditions (a)-(c). The following states obtained through our approach are not obtained in \cite{kmix}. Particularly, our state in Example \ref{95} boasts a higher purity compared with the one presented in \cite{kmix}.

\begin{example}
For the cases of $N=8,9$ and $k=2,3$, we obtain the states with purity $1$ \cite{npj}.
\end{example}

\begin{example}\label{74}
Case of $N=7$ and $k=4$: We can obtain the matrix,
$$
G_{4\times 7}=\left(\begin{array}{c}
0 0 1 1 1 1 1\\
0 1 0 2 2 2 2\\
1 0 2 0 2 3 3\\
2 2 0 1 2 1 3
\end{array}\right).
$$
Therefore, the $7$-qubit $4$-uniform mixed state can be described by $m=4$ generators:\\
\centerline{$G_1=\sigma_0\otimes\sigma_0\otimes\sigma_1\otimes\sigma_1\otimes\sigma_1\otimes\sigma_1\otimes\sigma_1,$}
\centerline{$G_2=\sigma_0\otimes\sigma_1\otimes\sigma_0\otimes\sigma_2\otimes\sigma_2\otimes\sigma_2\otimes\sigma_2,$}
\centerline{$G_3=\sigma_1\otimes\sigma_0\otimes\sigma_2\otimes\sigma_0\otimes\sigma_2\otimes\sigma_3\otimes\sigma_3,$}
\centerline{$G_4=\sigma_2\otimes\sigma_2\otimes\sigma_0\otimes\sigma_1\otimes\sigma_2\otimes\sigma_1\otimes\sigma_3.$}\\
A $4$-uniform $7$-qubit mixed state with a purity $\frac{1}{8}$ is generated as follows:
$$
\rho=\frac{1}{2^7}\sum_{j_1,\ldots,j_4=0}^{1}G_1^{j_1}\cdots G_4^{j_4}=\frac{1}{2^7}(\mathbb{I}+G_1)\cdots(\mathbb{I}+G_4),
$$
which is a symmetric mixture of the following eight pure states:\\
$|\varphi_1\rangle=
  |   0	0	0	0	0	0	0 \rangle +
i |   0	0	0	1	0	1	0 \rangle +
i |   0	0	1	0	1	0	1 \rangle +
  |   0	0	1	1	1	1	1 \rangle +
i |   0	1	0	0	1	0	1 \rangle +
  |   0	1	0	1	1	1	1 \rangle +
  |   0	1	1	0	0	0	0 \rangle +
i |   0	1	1	1	0	1	0 \rangle +
i |   1	0	0	0	0	0	1 \rangle -
  |  1	0	0	1	0	1	1 \rangle
- |  1	0	1	0	1	0	0 \rangle +
i |   1	0	1	1	1	1	0 \rangle +
  |   1	1	0	0	1	0	0 \rangle
-i|  1	1	0	1	1	1	0 \rangle
-i|  1	1	1	0	0	0	1 \rangle +
  |   1	1	1	1	0	1	1 \rangle,$

$|\varphi_2\rangle=
|   0	0	0	0	0	0	1 \rangle
-i |   0	0	0	1	0	1	1 \rangle
-i |   0	0	1	0	1	0	0 \rangle
+  |   0	0	1	1	1	1	0 \rangle
+i |   0	1	0	0	1	0	0 \rangle
-  |   0	1	0	1	1	1	0 \rangle
-  |   0	1	1	0	0	0	1 \rangle
+i |   0	1	1	1	0	1	1 \rangle
+i |   1	0	0	0	0	0	0 \rangle
+  |   1	0	0	1	0	1	0 \rangle
+  |   1	0	1	0	1	0	1 \rangle
+i |   1	0	1	1	1	1	1 \rangle
+  |   1	1	0	0	1	0	1 \rangle
+i |   1	1	0	1	1	1	1 \rangle
+i |   1	1	1	0	0	0	0 \rangle
+  |   1	1	1	1	0	1	0 \rangle,$

$|\varphi_3\rangle=
|   0	0	0	0	0	1	0 \rangle
+i |   0	0	0	1	0	0	0 \rangle
+i |   0	0	1	0	1	1	1 \rangle
+ |   0	0	1	1	1	0	1 \rangle
-i |   0	1	0	0	1	1	1 \rangle
- |   0	1	0	1	1	0	1 \rangle
- |   0	1	1	0	0	1	0 \rangle
-i |   0	1	1	1	0	0	0 \rangle
-i |   1	0	0	0	0	1	1 \rangle
+ |   1	0	0	1	0	0	1 \rangle
+ |   1	0	1	0	1	1	0 \rangle
-i |   1	0	1	1	1	0	0 \rangle
+ |   1	1	0	0	1	1	0 \rangle
-i |   1	1	0	1	1	0	0 \rangle
-i |   1	1	1	0	0	1	1 \rangle
+ |   1	1	1	1	0	0	1 \rangle,$

$|\varphi_4\rangle=
|   0	0	0	0	0	1	1 \rangle
-i |   0	0	0	1	0	0	1 \rangle
-i |   0	0	1	0	1	1	0 \rangle
+ |   0	0	1	1	1	0	0 \rangle
-i |   0	1	0	0	1	1	0 \rangle
+ |   0	1	0	1	1	0	0 \rangle
+ |   0	1	1	0	0	1	1 \rangle
-i |   0	1	1	1	0	0	1 \rangle
-i |   1	0	0	0	0	1	0 \rangle
- |   1	0	0	1	0	0	0 \rangle
- |   1	0	1	0	1	1	1 \rangle
-i |   1	0	1	1	1	0	1 \rangle
+ |   1	1	0	0	1	1	1 \rangle
+i |   1	1	0	1	1	0	1 \rangle
+i |   1	1	1	0	0	1	0 \rangle
+ |   1	1	1	1	0	0	0 \rangle,$

$|\varphi_5\rangle=
|   0	0	0	0	1	0	0 \rangle
-i |   0	0	0	1	1	1	0 \rangle
-i |   0	0	1	0	0	0	1 \rangle
+ |   0	0	1	1	0	1	1 \rangle
+i |   0	1	0	0	0	0	1 \rangle
- |   0	1	0	1	0	1	1 \rangle
- |   0	1	1	0	1	0	0 \rangle
+i |   0	1	1	1	1	1	0 \rangle
+i |   1	0	0	0	1	0	1 \rangle
+ |   1	0	0	1	1	1	1 \rangle
+ |   1	0	1	0	0	0	0 \rangle
+i |   1	0	1	1	0	1	0 \rangle
+ |   1	1	0	0	0	0	0 \rangle
+i |   1	1	0	1	0	1	0 \rangle
+i |   1	1	1	0	1	0	1 \rangle
+ |   1	1	1	1	1	1	1 \rangle,$

$|\varphi_6\rangle=
|   0	0	0	0	1	0	1\rangle
+i |   0	0	0	1	1	1	1\rangle
+i |   0	0	1	0	0	0	0\rangle
+ |   0	0	1	1	0	1	0\rangle
+i |   0	1	0	0	0	0	0\rangle
+ |   0	1	0	1	0	1	0\rangle
+ |   0	1	1	0	1	0	1\rangle
+i |   0	1	1	1	1	1	1\rangle
+i |   1	0	0	0	1	0	0\rangle
- |   1	0	0	1	1	1	0\rangle
- |   1	0	1	0	0	0	1\rangle
+i |   1	0	1	1	0	1	1\rangle
+ |   1	1	0	0	0	0	1\rangle
-i |   1	1	0	1	0	1	1\rangle
-i |   1	1	1	0	1	0	0\rangle
+ |   1	1	1	1	1	1	0\rangle,$

$|\varphi_7\rangle=
|   0	0	0	0	1	1	0 \rangle
-i |   0	0	0	1	1	0	0 \rangle
-i |   0	0	1	0	0	1	1 \rangle
+ |   0	0	1	1	0	0	1 \rangle
-i |   0	1	0	0	0	1	1 \rangle
+ |   0	1	0	1	0	0	1 \rangle
+ |   0	1	1	0	1	1	0 \rangle
-i |   0	1	1	1	1	0	0 \rangle
-i |   1	0	0	0	1	1	1 \rangle
- |   1	0	0	1	1	0	1 \rangle
- |   1	0	1	0	0	1	0 \rangle
-i |   1	0	1	1	0	0	0 \rangle
+ |   1	1	0	0	0	1	0 \rangle
+i |   1	1	0	1	0	0	0 \rangle
+i |   1	1	1	0	1	1	1 \rangle
+ |   1	1	1	1	1	0	1 \rangle,$

$|\varphi_8\rangle=
|   0	0	0	0	1	1	1 \rangle
+i |   0	0	0	1	1	0	1 \rangle
+i |   0	0	1	0	0	1	0 \rangle
+ |   0	0	1	1	0	0	0 \rangle
-i |   0	1	0	0	0	1	0 \rangle
- |   0	1	0	1	0	0	0 \rangle
- |   0	1	1	0	1	1	1 \rangle
-i |   0	1	1	1	1	0	1 \rangle
-i |   1	0	0	0	1	1	0 \rangle
+ |   1	0	0	1	1	0	0 \rangle
+ |   1	0	1	0	0	1	1 \rangle
-i |   1	0	1	1	0	0	1 \rangle
+ |   1	1	0	0	0	1	1 \rangle
-i |   1	1	0	1	0	0	1 \rangle
-i |   1	1	1	0	1	1	0 \rangle
+ |   1	1	1	1	1	0	0 \rangle.$
\end{example}

\begin{example}
Case of $N=8$ and $k=4$: We obtain
$$
G_{5\times 8}=\left(\begin{array}{c}
0	0	0	1	1	1	1	1\\
0	0	1	0	2	2	2	2\\
0	1	0	2	0	2	3	3\\
0	2	2	0	1	2	1	3\\
1	0	0	2	2	1	3	2
\end{array}\right),
$$
which leads to a $4$-uniform $8$-qubit mixed state with purity $\frac{1}{8}$.
\end{example}

\begin{example}
Case of $N=9$ and $k=4$: The $4$-uniform $9$-qubit mixed state is obtained from
$$G_{6\times 9}=\left(\begin{array}{c}
0	0	0	0	1	1	1	1	1\\
0	0	0	1	0	2	2	2	2\\
0	0	1	0	2	0	2	3	3\\
0	0	2	2	0	1	2	1	3\\
0	1	0	0	2	2	1	3	2\\
1	0	0	2	2	2	3	3	1
\end{array}\right).$$
The purity of this state is $\frac{1}{8}$.
\end{example}

\begin{example}
Case of $N=8$ and $k=5$: The $5$-uniform $8$-qubit mixed state is obtained from the following matrix,
$$G_{4\times 8}=\left(\begin{array}{c}
0	0	1	1	1	1	1	1\\
0	0	2	2	2	2	2	2\\
1	1	0	0	1	1	2	2\\
2	2	0	0	2	2	3	3
\end{array}\right).$$
The state has a purity $\frac{1}{16}$.
\end{example}

\begin{example}\label{95}
Case of $N=9$ and $k=5$: We have the following matrix
$$G_{6\times 9}=\left(\begin{array}{c}
0	0	0	1	1	1	1	1	1\\
0	0	0	2	2	2	2	2	2\\
0	1	1	0	0	1	1	2	2\\
0	2	2	0	0	2	2	3	3\\
1	0	1	0	3	0	2	2	3\\
2	0	2	0	1	0	3	3	1
\end{array}\right),$$
which gives rise to a $5$-uniform $9$-qubit mixed state $\rho$ with purity $\frac{1}{8}$, $\rho=\frac{1}{2^9}\sum_{j_1,\ldots,j_6=0}^{1}G_1^{j_1}\cdots G_6^{j_6}=\frac{1}{2^9}(\mathbb{I}+G_1)\cdots(\mathbb{I}+G_6)$, where\\
\centerline{$G_1=\sigma_0\otimes\sigma_0\otimes\sigma_0\otimes\sigma_1\otimes\sigma_1\otimes\sigma_1\otimes\sigma_1\otimes\sigma_1\otimes\sigma_1,$}
\centerline{$G_2=\sigma_0\otimes\sigma_0\otimes\sigma_0\otimes\sigma_2\otimes\sigma_2\otimes\sigma_2\otimes\sigma_2\otimes\sigma_2\otimes\sigma_2,$}
\centerline{$G_3=\sigma_0\otimes\sigma_1\otimes\sigma_1\otimes\sigma_0\otimes\sigma_0\otimes\sigma_1\otimes\sigma_1\otimes\sigma_2\otimes\sigma_2,$}
\centerline{$G_4=\sigma_0\otimes\sigma_2\otimes\sigma_2\otimes\sigma_0\otimes\sigma_0\otimes\sigma_2\otimes\sigma_2\otimes\sigma_3\otimes\sigma_3.$}
\centerline{$G_5=\sigma_1\otimes\sigma_0\otimes\sigma_1\otimes\sigma_0\otimes\sigma_3\otimes\sigma_0\otimes\sigma_2\otimes\sigma_2\otimes\sigma_3,$}
\centerline{$G_6=\sigma_2\otimes\sigma_0\otimes\sigma_2\otimes\sigma_0\otimes\sigma_1\otimes\sigma_0\otimes\sigma_3\otimes\sigma_3\otimes\sigma_1.$}
\end{example}

\begin{example}
Case of $N=9$ and $k=6$: The matrix
$$G_{3\times 9}=\left(\begin{array}{c}
0	0	1	1	1	1	1	1	1\\
0	1	0	2	2	2	2	2	2\\
1	2	2	0	0	1	2	3	3
\end{array}\right)$$
satisfies the three conditions of Theorem \ref{gsan} and generates a $6$-uniform $9$-qubit mixed state with purity $\frac{1}{2^6}$.
\end{example}

\section{$k$-Uniform Mixed States Based on Orthogonal Arrays}\label{gw}

In \cite{kmix}, K{\l}obus et al. derived $k$-uniform states using a specific set of generators. In this section, we construct such states with the help of OAs. An {\rm OA}$(r,N,d,k)$ of strength $k$ is an $r \times N$ matrix with entries from $d$ symbols, characterized by the property that any $r \times k$ submatrix contains all possible combinations of $k$ symbols with equal frequency. The OA is called irredundant, denoted as IrOA, if when one removes from the array any $k$ columns, all the remaining $r$ rows, containing $N-k$ symbols each, are different. The Hamming distance between two row vectors in a matrix $A$ is defined to be the number of positions at which they differ. The minimal Hamming distance of $A$, denoted as ${\rm MD}(A)$, is the smallest Hamming distance between any pair of distinct rows \cite{npj}. Suppose that the rows of an OA$(r,N,d,k)$ can be partitioned into $m$ submatrices $\{A_1, A_2, \ldots, A_{m}\}$ such that each $A_i$ is an OA$(r/m,N,d,k_1)$ with $k_1\geq0$. Then the set $\{A_1, A_2, \ldots, A_{m}\}$ is called an orthogonal partition of strength $k_1$ with $m$ blocks of the OA$(r,N,d,k)$ \cite{pang21w}. Utilizing the orthogonal partitions of OAs, we propose a general method for constructing arbitrary dimensional $k$-uniform pure or mixed states, enabling the creation of infinite families of such states. In this study, we do not consider OAs that contain two identical rows to simplify the purity calculations of the derived states. Notably, for $k=1$, the following method consistently produces 1-uniform pure states.

\begin{theorem}\label{1un}
A $1$-uniform pure state of $N$ qudits exists for $N\geq 2$.
\end{theorem}

\begin{proof}
OAs of strength 1 can be constructed trivially. From Ref. \cite{npj} an ${\rm OA}(d,N,d,1)$ exists with the minimal Hamming distance $N-1+1=N\geq2$. Therefore, it is irredundant and corresponds to a 1-uniform pure state of $N$ qudits.
\end{proof}

For $k\geq2$, we present a general construction method of these $k$-uniform states as follows:

\begin{theorem}\label{th}
If the set $\{A_1, A_2, \ldots, A_{m}\}$ is an orthogonal partition of an ${\rm OA}(r,N,d,k)$ and $\text{MD}(A_i) \geq k+1$ for any $i=1,2,\ldots,m$, then when $m=1$, the OA is irredundant and it corresponds to a $k$-uniform pure state. When $m \geq 2$, the resulting $k$-uniform $N$-qudit mixed state given by the symmetric mixture of $m$ pure states from $A_i$ and its purity is $\frac{1}{m}$. In particular, an {\rm OA}$(r,N,d,k)$ leads to an $N$-qudit $k$-uniform mixed state given by the symmetric mixture of $r$ pure states and with purity $\frac{1}{r}$.
\end{theorem}

\begin{proof}
Let $|\phi_i\rangle$, $i=1,2,\ldots,m$, be $m$ pure states obtained from each $A_i$ of an $N$-qudit system. The associated density matrix $\rho$ reads, $\rho=\frac{1}{m}\sum_{i=1}^m|\phi_i\rangle\langle\phi_i|$. We divide the system into two parts $\mathcal{S}_A$ and $\mathcal{S}_B$, each
containing $N_A=k$ and $N_B=N-k$ qudits, respectively. To determine the density matrix associated with $\mathcal{S}_A$, we consider the reduced state $\rho_A={\rm Tr}_B(\rho)$.

When $m=1$, OA is irredundant and it corresponds to a $k$-uniform pure state \cite{npj}.

When $m\geq2$, we consider the mixed states with maximally mixed reductions, namely, $\rho_A$ is proportional to the identity matrix. Since ${\rm MD}(A_i)\geq k+1$, the sequence of $N_B$ symbols that appears in every row of a subset of $N_B$ columns is unique across each $A_i$. This uniqueness renders the reduced density matrix $\rho_A$ diagonal. In addition, the sequence of $N_A$ symbols appearing in every row of every subset of $N_A$ columns of {\rm OA}$(r,N,d,k)$ is repeated an equal number of times. Consequently, the diagonal entries of $\rho_A$ are all equal.

Since the OA does not contain two identical rows, we have $\langle\phi_i|\phi_j\rangle=0$ for $i\neq j\in \{1,2,\ldots,m\}$. Moreover, as $|\phi_i\rangle$ is a pure state, we also have $\langle\phi_i|\phi_i\rangle=1$ and ${\rm Tr}(\rho_i)={\rm Tr}(|\phi_i\rangle\langle\phi_i|)=1$ for $i\in \{1,2,\ldots,m\}$. Therefore,
\begin{align*}
\rho^2&=(\frac{1}{m}\sum_{i=1}^m|\phi_i\rangle\langle\phi_i|)
(\frac{1}{m}\sum_{j=1}^m|\phi_j\rangle\langle\phi_j|)\\
&=\frac{1}{m^2}\sum_{i=1}^m\sum_{j=1}^m|\phi_i\rangle\langle\phi_i|\cdot|\phi_j\rangle\langle\phi_j|\\
&=\frac{1}{m^2}\sum_{i=1}^m\sum_{j=1}^m|\phi_i\rangle\langle\phi_i|\phi_j\rangle\langle\phi_j|\\
&=\frac{1}{m^2}\sum_{i=1}^m|\phi_i\rangle\langle\phi_i|\phi_i\rangle\langle\phi_i|\\
&=\frac{1}{m^2}\sum_{i=1}^m|\phi_i\rangle\langle\phi_i|.
\end{align*}
Therefore, ${\rm Tr}(\rho^2)={\rm Tr}(\frac{1}{m^2}\sum_{i=1}^m|\phi_i\rangle\langle\phi_i|)=\frac{1}{m}$ and the purity of the state is $\frac{1}{m}$.
\end{proof}

By using Theorem \ref{th}, we can obtain a $k$-uniform pure or mixed states from a suitable OA$(r,N,d,k)$.

\begin{example}
Case of $d=2$, $N=5$ and $k=4$: Consider
$$
A={\rm OA}(16,5,2,4)=
\left( \begin{array}{c}
0 0 0 0 0\\
0 0 0 1 1\\
0 0 1 0 1\\
0 0 1 1 0\\
0 1 0 0 1\\
0 1 0 1 0\\
0 1 1 0 0\\
0 1 1 1 1\\
1 0 0 0 1\\
1 0 0 1 0\\
1 0 1 0 0\\
1 0 1 1 1\\
1 1 0 0 0\\
1 1 0 1 1\\
1 1 1 0 1\\
1 1 1 1 0\\
\end{array}
\right).
$$
Then $A$ yields a $2$-uniform $5$-qubit mixed state that can be expressed as a symmetric mixture of the following $16$ pure states,
$|\Phi_{1}\rangle=
|00000\rangle,$
$|\Phi_{2}\rangle=
|00011\rangle,$
$|\Phi_{3}\rangle=
|00101\rangle,$
$|\Phi_{4}\rangle=
|00110\rangle,$
$|\Phi_{5}\rangle=
|01001\rangle,$
$|\Phi_{6}\rangle=
|01010\rangle,$
$|\Phi_{7}\rangle=
|01100\rangle,$
$|\Phi_{8}\rangle=
|01111\rangle,$
$|\Phi_{9}\rangle=
|10001\rangle,$
$|\Phi_{10}\rangle=
|10010\rangle,$
$|\Phi_{11}\rangle=
|10100\rangle,$
$|\Phi_{12}\rangle=
|10111\rangle,$
$|\Phi_{13}\rangle=
|11000\rangle,$
$|\Phi_{14}\rangle=
|11011\rangle,$
$|\Phi_{15}\rangle=
|11101\rangle,$
$|\Phi_{16}\rangle=
|11110\rangle.$ The purity of this state is $\frac{1}{16}$, which attains the highest purity \cite{kmix}.
\end{example}

\begin{example}\label{273}
Case of $d=2$, $N=7$ and $k=3$: Let
$$A_1={\rm OA}(8,7,2,2)=\left( \begin{array}{c}
0 0 0 0 0 0 0\\
0 0 0 1 1 1 1\\
0 1 1 0 0 1 1\\
0 1 1 1 1 0 0\\
1 0 1 0 1 0 1\\
1 0 1 1 0 1 0\\
1 1 0 0 1 1 0\\
1 1 0 1 0 0 1
\end{array}
\right)$$ and
$$A_2=1\oplus A_1=\left( \begin{array}{c}
1 1 1 1 1 1 1\\
1 1 1 0 0 0 0\\
1 0 0 1 1 0 0\\
1 0 0 0 0 1 1\\
0 1 0 1 0 1 0\\
0 1 0 0 1 0 1\\
0 0 1 1 0 0 1\\
0 0 1 0 1 1 0\\
\end{array}
\right)
$$
where $\oplus$ denotes the Kronecker sum \cite{npj}. Then $\left(\begin{array}{c}
A_1\\
A_2
\end{array}\right)$ is an ${\rm OA}(16,7,2,3)$ \cite{hss}. Since ${\rm MD}(A_1)={\rm MD}(A_2)=4$ \cite{muk95}, according to Theorem \ref{th} we have the state from the symmetric mixture of two pure states $|\Phi_{1}\rangle=\frac{1}{2\sqrt{2}}(
|0000000\rangle+
|0001111\rangle+
|0110011\rangle+
|0111100\rangle+
|1010101\rangle+
|1011010\rangle+
|1100110\rangle+
|1101001\rangle)$ and
$|\Phi_{2}\rangle=\frac{1}{2\sqrt{2}}(|11111111\rangle+
|1110000\rangle+
|1001100\rangle+
|1000011\rangle+
|0101010\rangle+
|0100101\rangle+
|0011001\rangle+
|0010110\rangle)$. The purity of this $3$-uniform $7$-qubit mixed state is $\frac{1}{2}$.
\end{example}

Note that for the case of $d=2$, $N=5$ and $k=4$ in Example 8 ($d=2$, $N=7$ and $k=3$ in Example 9), one may also obtain such $k$-uniform states in terms of the approach in \cite{kmix}. Nevertheless, our approach provides states whose coefficients are all $1$, which are simpler than the ones obtained via the approach in \cite{kmix}.

Let $\mathcal{A}$ be an abelian group of order $d$. $\mathcal{A}^k$, $k\geq1$, denotes the additive group of order $d^k$ consisting of all $k$-tuples of entries from $\mathcal{A}$ with the usual vector addition as the binary operation. Let $\mathcal{A}_0^k=\{(x_1,x_2,\ldots,x_k):x_1=\cdots=x_k\in\mathcal{A}\}$. Accordingly, $\mathcal{A}_0^k$ is a subgroup of $\mathcal{A}^k$ of order $d$, with its cosets denoted by $\mathcal{A}_i^k$, $i=1,\ldots,d^{k-1}-1$. An $r\times N$ array $D$ based on $\mathcal{A}$ is a difference scheme of strength $k$ if, for every $r\times k$ subarray, each set $\mathcal{A}_i^k$, $i=0,1,\ldots,d^{k-1}-1$, is represented equally often when the rows of the subarray are viewed as elements of $\mathcal{A}^k$. Such an array is denoted by $D_k(r,N,d)$. The difference schemes are simple yet effective tools in constructing OAs whose Hamming distances possess specific properties.

\begin{theorem}\label{cz}
A $k$-uniform $N$-qudit mixed state with purity $\frac{1}{r}$ can be obtained from a difference scheme $D_k(r,N,d)$ with $k<N$,
\end{theorem}

\begin{proof}
Let $$D_k(r,N,d)=\left(\begin{array}{c}
{\bf a}_1\\
{\bf a}_2\\
\vdots\\
{\bf a}_r\\
\end{array}\right)$$ and $A_i={\bf a}_i\oplus {\bf (d)}$ where ${\bf (d)}=(0,1,\ldots,d-1)^T$ for $i=1,2,\ldots,r$. Obviously, {\rm MD}$(A_i)=N\geq k+1$ and $D_k(r,N,d)\oplus {\bf (d)}$ is an OA. From every $A_i$, we obtain an $N$-qudit pure state with $d$ terms. Therefore, by Theorem \ref{th}, the symmetric mixture of these $r$ pure states yields a $k$-uniform mixed state of $N$ qudits with purity $\frac{1}{r}$.
\end{proof}

From the existing difference schemes and Theorem \ref{cz}, we have the following results.

\begin{corollary}\label{kjia1}
(1) There exist $3$-uniform $4$-qudit mixed states with purity $\frac{1}{d^2}$ for any integer $d\geq2$.

(2) If $d\neq 2\, (mod \ 4)$, there exist $k$-uniform $k+1$-qudit mixed states with purity $\frac{1}{d^{k-1}}$.

(3) If $d\leq k$ and odd $k$, there exist $k$-uniform $k+1$-qudit mixed states with purity $\frac{1}{d^{k-1}}$.

(4) If $d$ is an odd prime power, there exist $3$-uniform $d+1$-qudit mixed states with purity $\frac{1}{d^2}$. If $d$ is an even prime power, there exist $3$-uniform $d+2$-qudit mixed states with purity $\frac{1}{d^2}$.
\end{corollary}

\begin{proof}
From the existence of difference schemes \cite{chen17,zx22}, we have infinitely uniform mixed states by Theorem \ref{cz}.

(1) There exists a difference scheme $D_3(d^{2},4,d)$ for any $d$.

(2) There exists a difference scheme $D_k(d^{k-1},k+1,d)$ for $d\neq 2 (mod \ 4)$.

(3) If $d\leq k$ and $k$ is odd, a difference scheme $D_k(d^{k-1},k+1,d)$ exists.

(4) If $d$ is an odd prime power, a difference scheme $D_3(d^2,d+1,d)$ exists. If $d$ is an even prime power, a difference scheme $D_3(d^2,d+2,d)$ exists.
\end{proof}

\begin{example}\label{zg}
From Corollary \ref{kjia1} (3), $3$-uniform $4$-qubit mixed state with purity $\frac{1}{4}$ and $4$-uniform $5$-qubit mixed state with purity $\frac{1}{16}$ exist, all of which attain the highest purity \cite{kmix}.
\end{example}

\begin{example}\label{18}
From a difference scheme $D_3(18,5,3)$ \cite{hjg,h97} and Theorem \ref{cz}, we obtain a $3$-uniform mixed state of $5$-qutrit with purity $\frac{1}{18}$.
\end{example}

\begin{theorem}\label{dx}
If there exists an ${\rm OA}(r,N,d,k)$ with a minimal Hamming distance of $k+1$, we obtain an $(N-x)$-qudit $k$-uniform mixed state. The state can be expressed as a symmetric mixture of $d^x$ pure states for $1\leq x\leq k$ and $N-x>k$. The purity of the state is $\frac{1}{d^x}$.
\end{theorem}

\begin{proof}
If an ${\rm OA}(r,N,d,k)$ with a minimal Hamming distance of $k+1$ exists, then taking the runs in the OA that begin with 0 (or any other particular symbol) and omitting the first column, we obtain an ${\rm OA}(\frac{r}{d},N-1,d,k-1)$ with minimal Hamming distance $\geq k+1$. Then there exists an orthogonal partition of strength $k-1$ of the last $N-1$ columns of ${\rm OA}(r,N,d,k)$. Therefore, by Theorem \ref{th} we obtain an $(N-1)$-qudit $k$-uniform mixed state that can be expressed as a symmetric mixture of $d$ pure states. The purity of the state is $\frac{1}{d}$.

Similarly, considering the runs in the ${\rm OA}(r,N,d,k)$ that begin with $(0,0)$ (or any other particular symbols) and omitting the first two columns, we have an ${\rm OA}(\frac{r}{d^2},N-2,d,k-2)$ with minimal Hamming distance greater or equal $k+1$. Therefore, we obtain an $(N-2)$-qudit $k$-uniform mixed state from the symmetric mixture of $d^2$ pure states by omitting the first two columns. The other cases can be proved similarly.
\end{proof}

\begin{corollary}\label{35}
If ${\text gcd}\{d,4\}\neq2$ and ${\text gcd}\{d,18\}\neq3$ where ${\text gcd}$ refers to the greatest common divisor, a $5$-qudit $3$-uniform mixed states exists and its purity is $\frac{1}{d}$.
\end{corollary}

\begin{proof}
If ${\text gcd}\{d,4\}\neq2$ and ${\text gcd}\{d,18\}\neq3$, an ${\rm OA}(d^3,6,d,3)$ \cite{yin10} has minimal Hamming distance $4$ from Ref. \cite{npj}. Accordingly, the result can be proved by Theorem \ref{dx}.
\end{proof}

\begin{example}\label{64}
Case of $d=4$, $k=3$ and $N=4,5$: We can obtain a $3$-uniform mixed state of
$4$-ququart with purity $\frac{1}{16}$ and a $3$-uniform mixed state of $5$-ququart with purity $\frac{1}{4}$.

A difference scheme $D_3(16,6,4)$ exists \cite{zx22,hjg} and can be written as
$$[{\bf 0}_{16},A]=\left(\begin{array}{ccc}
&{\bf a}_1&a_{1,5}\\
{\bf 0}_{16}&\vdots&\vdots\\
&{\bf a}_{16}&a_{16,5}
\end{array}\right)=\left(\begin{array}{c}
0	0	0	0	0	0\\
0	0	1	2	1	2\\
0	0	2	3	2	3\\
0	0	3	1	3	1\\
0	1	0	1	2	2\\
0	1	1	3	3	0\\
0	1	2	2	0	1\\
0	1	3	0	1	3\\
0	2	0	2	3	3\\
0	2	1	0	2	1\\
0	2	2	1	1	0\\
0	2	3	3	0	2\\
0	3	0	3	1	1\\
0	3	1	1	0	3\\
0	3	2	0	3	2\\
0	3	3	2	2	0
\end{array}\right),
$$
where ${\bf 0}_{16}$ is a $16\times1$ vector with all elements equal to zero and ${\bf a}_i$ is a $1\times4$ vector for $i=1,2,\ldots,16$. Therefore $D_3(16,6,4)\oplus {\bf (4)}$ is an ${\rm IrOA}(64,6,4,3)$ corresponding to a $3$-uniform pure state of $6$-ququart. Let
$$
A_1={\bf a}_1\oplus {\bf (4)}=\left(\begin{array}{c}
0000\\
1111\\
2222\\
3333
\end{array}\right),
$$
$$A_2={\bf a}_2\oplus {\bf (4)}=\left(\begin{array}{c}
0121\\
1030\\
2303\\
3212
\end{array}\right),
\ldots,$$
and
$$
A_{16}={\bf a}_{16}\oplus {\bf (4)}=\left(\begin{array}{c}
3322\\
2233\\
1100\\
0011
\end{array}\right).
$$
${\rm MD}(A_1)=\cdots={\rm MD}(A_{16})=4$. Therefore, we have a $3$-uniform $4$-ququart mixed state that can be expressed as a symmetric mixture of the following $16$ pure states,\\
$|\Phi_{1}\rangle=\frac{1}{2}(
|0000\rangle+
|1111\rangle+
|2222\rangle+
|3333\rangle),$\\
$|\Phi_{2}\rangle=\frac{1}{2}(
|	0	1	2	1	\rangle+
|	1	0	3	0	\rangle+
|	2	3	0	3	\rangle+
|	3	2	1	2	\rangle),$\\
$|\Phi_{3}\rangle=\frac{1}{2}(
|	0	2	3	2	\rangle+
|	1	3	2	3	\rangle+
|	2	0	1	0	\rangle+
|	3	1	0	1	\rangle),$\\
$|\Phi_{4}\rangle=\frac{1}{2}(
|	0	3	1	3\rangle+
|	1	2	0	2\rangle+
|	2	1	3	1\rangle+
|	3	0	2	0\rangle),$\\
$|\Phi_{5}\rangle=\frac{1}{2}(
|	1	0	1	2\rangle+
|	0	1	0	3\rangle+
|	3	2	3	0\rangle+
|	2	3	2	1\rangle),$\\
$|\Phi_{6}\rangle=\frac{1}{2}(
|	1	1	3	3\rangle+
|	0	0	2	2\rangle+
|	3	3	1	1\rangle+
|	2	2	0	0\rangle),$\\
$|\Phi_{7}\rangle=\frac{1}{2}(
|	1	2	2	0\rangle+
|	0	3	3	1\rangle+
|	3	0	0	2\rangle+
|	2	1	1	3\rangle),$\\
$|\Phi_{8}\rangle=\frac{1}{2}(
|	1	3	0	1\rangle+
|	0	2	1	0\rangle+
|	3	1	2	3\rangle+
|	2	0	3	2\rangle),$\\
$|\Phi_{9}\rangle=\frac{1}{2}(
|	2	0	2	3\rangle+
|	3	1	3	2\rangle+
|	0	2	0	1\rangle+
|	1	3	1	0\rangle),$\\
$|\Phi_{10}\rangle=\frac{1}{2}(
|	2	1	0	2\rangle+
|	3	0	1	3\rangle+
|	0	3	2	0\rangle+
|	1	2	3	1\rangle),$\\
$|\Phi_{11}\rangle=\frac{1}{2}(
|	2	2	1	1\rangle+
|	3	3	0	0\rangle+
|	0	0	3	3\rangle+
|	1	1	2	2\rangle),$\\
$|\Phi_{12}\rangle=\frac{1}{2}(
|	2	3	3	0\rangle+
|	3	2	2	1\rangle+
|	0	1	1	2\rangle+
|	1	0	0	3\rangle),$\\
$|\Phi_{13}\rangle=\frac{1}{2}(
|	3	0	3	1\rangle+
|	2	1	2	0\rangle+
|	1	2	1	3\rangle+
|	0	3	0	2\rangle),$\\
$|\Phi_{14}\rangle=\frac{1}{2}(
|	3	1	1	0\rangle+
|	2	0	0	1\rangle+
|	1	3	3	2\rangle+
|	0	2	2	3\rangle),$\\
$|\Phi_{15}\rangle=\frac{1}{2}(
|	3	2	0	3\rangle+
|	2	3	1	2\rangle+
|	1	0	2	1\rangle+
|	0	1	3	0\rangle),$\\
$|\Phi_{16}\rangle=\frac{1}{2}(
|	3	3	2	2\rangle+
|	2	2	3	3\rangle+
|	1	1	0	0\rangle+
|	0	0	1	1\rangle).$\\
The purity of this $3$-uniform mixed state is $\frac{1}{16}$.

Similarly, by Theorem \ref{dx}, $A\oplus {\bf (4)}$ is an ${\rm OA}(64,5,4,3)$ and {\rm MD}$(A)=4$. From $A_1=A$, $A_2=1\oplus A$, $A_3=2\oplus A$ and $A_4=3\oplus A$, we obtain a $3$-uniform mixed state of $5$-ququart with purity $\frac{1}{4}$, which is given by the symmetric mixture of the following four pure states:\\ $|\Phi_{1}\rangle=\frac{1}{4}(
|0	0	0	0	0\rangle+
|0	1	2	1	2\rangle+
|0	2	3	2	3\rangle+
|0	3	1	3	1\rangle+
|1	0	1	2	2\rangle+
|1	1	3	3	0\rangle+
|1	2	2	0	1\rangle+
|1	3	0	1	3\rangle+
|2	0	2	3	3\rangle+
|2	1	0	2	1\rangle+
|2	2	1	1	0\rangle+
|2	3	3	0	2\rangle+
|3	0	3	1	1\rangle+
|3	1	1	0	3\rangle+
|3	2	0	3	2\rangle+
|3	3	2	2	0\rangle),$\\
$|\Phi_{2}\rangle=\frac{1}{4}(
|1	1	1	1	1\rangle+
|1	0	3	0	3\rangle+
|1	3	2	3	2\rangle+
|1	2	0	2	0\rangle+
|0	1	0	3	3\rangle+
|0	0	2	2	1\rangle+
|0	3	3	1	0\rangle+
|0	2	1	0	2\rangle+
|3	1	3	2	2\rangle+
|3	0	1	3	0\rangle+
|3	3	0	0	1\rangle+
|3	2	2	1	3\rangle+
|2	1	2	0	0\rangle+
|2	0	0	1	2\rangle+
|2	3	1	2	3\rangle+
|2	2	3	3	1\rangle),$\\
$|\Phi_{3}\rangle=\frac{1}{4}(
|2	2	2	2	2\rangle+
|2	3	0	3	0\rangle+
|2	0	1	0	1\rangle+
|2	1	3	1	3\rangle+
|3	2	3	0	0\rangle+
|3	3	1	1	2\rangle+
|3	0	0	2	3\rangle+
|3	1	2	3	1\rangle+
|0	2	0	1	1\rangle+
|0	3	2	0	3\rangle+
|0	0	3	3	2\rangle+
|0	1	1	2	0\rangle+
|1	2	1	3	3\rangle+
|1	3	3	2	1\rangle+
|1	0	2	1	0\rangle+
|1	1	0	0	2\rangle),$\\
$|\Phi_{4}\rangle=\frac{1}{4}(
|3	3	3	3	3\rangle+
|3	2	1	2	1\rangle+
|3	1	0	1	0\rangle+
|3	0	2	0	2\rangle+
|2	3	2	1	1\rangle+
|2	2	0	0	3\rangle+
|2	1	1	3	2\rangle+
|2	0	3	2	0\rangle+
|1	3	1	0	0\rangle+
|1	2	3	1	2\rangle+
|1	1	2	2	3\rangle+
|1	0	0	3	1\rangle+
|0	3	0	2	2\rangle+
|0	2	2	3	0\rangle+
|0	1	3	0	1\rangle+
|0	0	1	1	3\rangle).$\\
\end{example}

\begin{example}\label{34}
Case of $d=3$, $k=4$ and $6\leq N\leq9$: As there exists a $(11,3^5,6)_3$ Golay code \cite{hss}, an ${\rm OA}(3^5,11,3,4)$ with minimal Hamming distance $6$ can be obtained. Hence, we obtain an ${\rm OA}(3^5,10,3,4)$ with minimal Hamming distance $5$. Therefore, the OA is an ${\rm IrOA}(3^5,10,3,4)$. By Theorem \ref{dx}, we obtain a $(10-x)$-qutrit $4$-uniform mixed state, which can be expressed as a symmetric mixture of $3^x$ pure states for $1\leq x\leq 4$. The purity of the state is $\frac{1}{3^x}$.
\end{example}

\begin{example}\label{351}
Case of $d=3$, $k=5$ and $7\leq N\leq11$: In this case, a $(12,3^6,6)_3$ self-dual Golay code exists\cite{hss}. The corresponding OA is an ${\rm OA}(3^6,12,3,5)$ with minimal Hamming distance $6$. Similarly, we obtain a 5-uniform $(12-x)$-qutrit mixed state with purity $\frac{1}{3^x}$ for $1\leq x\leq 5$.
\end{example}

\begin{example}\label{44}
Case of $d=4$, $k=4$ and $6\leq N\leq9$: An ${\rm OA}(4^5,11,4,4)$ with minimal Hamming distance $6$ exists\cite{hss}. Therefore, we obtain a $4$-uniform $(10-x)$-ququart mixed state with purity $\frac{1}{4^x}$ for $1\leq x\leq 4$.
\end{example}

\begin{example}\label{45}
Case of $d=4$, $k=5$ and $7\leq N\leq11$: There exists a $(12,4^6,6)_4$ quadratic residue code \cite{hss}. The corresponding OA is an ${\rm OA}(4^6,12,4,5)$ with minimal Hamming distance $6$. Therefore, we obtain a $5$-uniform $(12-x)$-ququart mixed state with purity $\frac{1}{4^x}$ for $1\leq x\leq 5$.
\end{example}

\begin{example}\label{46}
Case of $d=4$, $k=6$ and $N=9$: From Ref. \cite{pang21w} and an {\rm OA}$(16,4,4,2)$, there exists an {\rm OA}$(256,4,4,4)$ with an orthogonal partition of strength $2$, namely, $\{A_1, A_2, \ldots, A_{16}\}$. Since an {\rm OA}$(64,6,4,3)$ can be represented as
$$\left(\begin{array}{cc}
{\mathbf 0}_{16}&B_1\\
{\mathbf 1}_{16}&B_2\\
{\mathbf 2}_{16}&B_3\\
{\mathbf 3}_{16}&B_4
\end{array}\right),$$
the set $\{B_1,B_2,B_3,B_4\}$ is an orthogonal partition of strength $2$ of $B={\rm OA}(64,5,4,3)$. Therefore, from Ref. \cite{pang21w} and $B$, an {\rm OA}$(4^5,5,4,5)$ with an orthogonal partition of strength $3$, namely, $\{C_1, C_2, \ldots, C_{16}\}$ exists.

Therefore, it also follows from Ref. \cite{pang21w} that
\begin{align*}
M=&\left(\begin{array}{cc}
A_{1}\otimes \bm{1}_{64}&\bm{1}_{16}\otimes C_{1}\\
A_{2}\otimes \bm{1}_{64}&\bm{1}_{16}\otimes C_{2}\\
\vdots&\vdots\\
A_{16}\otimes \bm{1}_{64}&\bm{1}_{16}\otimes C_{16}\\
\end{array}\right)
\end{align*} is an ${\rm OA}(4^7,9,4,6)$.

Assume $A_1=\left(\begin{array}{c}
\mathbf{a}_{1}\\
\mathbf{a}_{2}\\
\vdots\\
\mathbf{a}_{16}\\
\end{array}\right)$ and $C_1=\left(\begin{array}{c}
B_1\\
B_2\\
B_3\\
B_4\\
\end{array}\right)=\left(\begin{array}{c}
\mathbf{c}_{1}\\
\mathbf{c}_{2}\\
\vdots\\
\mathbf{c}_{64}\\
\end{array}\right)$.
Then $[A_{1}\otimes\bm{1}_{64},\bm{1}_{16}\otimes C_{1}]=\left(\begin{array}{cc}
\mathbf{a}_{1}&\mathbf{c}_{1}\\
\mathbf{a}_{1}&\mathbf{c}_{2}\\
\vdots&\vdots\\
\mathbf{a}_{1}&\mathbf{c}_{64}\\
\mathbf{a}_{2}&\mathbf{c}_{1}\\
\vdots&\vdots\\
\mathbf{a}_{2}&\mathbf{c}_{64}\\
\vdots&\vdots\\
\mathbf{a}_{16}&\mathbf{c}_{1}\\
\vdots&\vdots\\
\mathbf{a}_{16}&\mathbf{c}_{64}\\
\end{array}\right)$.

Set $$E_1=\left(\begin{array}{cc}
\mathbf{a}_{1}&\mathbf{c}_{1}\\
\mathbf{a}_{2}&\mathbf{c}_{2}\\
\vdots&\vdots\\
\mathbf{a}_{16}&\mathbf{c}_{16}\\
\end{array}\right),$$
$$E_2=\left(\begin{array}{cc}
\mathbf{a}_{1}&\mathbf{c}_{2}\\
\mathbf{a}_{2}&\mathbf{c}_{3}\\
\vdots&\vdots\\
\mathbf{a}_{16}&\mathbf{c}_{1}\\
\end{array}\right),$$
$$E_3=\left(\begin{array}{cc}
\mathbf{a}_{1}&\mathbf{c}_{3}\\
\mathbf{a}_{2}&\mathbf{c}_{4}\\
\vdots&\vdots\\
\mathbf{a}_{16}&\mathbf{c}_{2}\\
\end{array}\right),$$
$$\vdots$$
$$E_{16}=\left(\begin{array}{cc}
\mathbf{a}_{1}&\mathbf{c}_{16}\\
\mathbf{a}_{2}&\mathbf{c}_{1}\\
\vdots&\vdots\\
\mathbf{a}_{16}&\mathbf{c}_{15}\\
\end{array}\right),$$
$$E_{17}=\left(\begin{array}{cc}
\mathbf{a}_{1}&\mathbf{c}_{17}\\
\mathbf{a}_{2}&\mathbf{c}_{18}\\
\vdots&\vdots\\
\mathbf{a}_{16}&\mathbf{c}_{32}\\
\end{array}\right),$$
$$\vdots,$$
$$E_{63}=\left(\begin{array}{cc}
\mathbf{a}_{1}&\mathbf{c}_{63}\\
\mathbf{a}_{2}&\mathbf{c}_{64}\\
\vdots&\vdots\\
\mathbf{a}_{16}&\mathbf{c}_{62}\\
\end{array}\right),$$
$$E_{64}=\left(\begin{array}{cc}
\mathbf{a}_{1}&\mathbf{c}_{64}\\
\mathbf{a}_{2}&\mathbf{c}_{49}\\
\vdots&\vdots\\
\mathbf{a}_{16}&\mathbf{c}_{63}\\
\end{array}\right).$$
Then $[A_{1}\otimes \bm{1}_{64},\bm{1}_{16}\otimes C_{1}]$ can be expressed as $$\left(\begin{array}{cc}
E_1\\
E_2\\
\vdots\\
E_{64}\\
\end{array}\right).$$
Since $E_i=[{\rm OA}(16,4,4,2),{\rm OA}(16,5,4,2)]$, $i=1,2,\ldots,64$, we have ${\rm MD}(E_i)=7$. Similarly, the matrix $[A_{j}\otimes \bm{1}_{64},\bm{1}_{16}\otimes C_{j}]$ for each $j=2,3,\ldots,16$, can also be orthogonally partitioned into $64$ submatrices with the form $[{\rm OA}(16,4,4,2),{\rm OA}(16,5,4,2)]$. Therefore, by Theorem \ref{th} we obtain a $9$-ququart $6$-uniform mixed state with purity $\frac{1}{4^5}$.
\end{example}

\section{Discussions on Purity of $k$-Uniform States}\label{dis}

In \cite{kmix} K{\l}obus et al. conjectured the possibility of constructing $k$-uniform states with high purity and provided numerical evidence that $k$-uniform states attain the highest purity for $k\leq6$. In this study, based on our construction method we have obtained states with the highest possible purity, which answers affirmatively the open question raised in Ref. \cite{kmix}: what is the highest possible purity of a $k$-uniform state for a given number of parties $N$?

Interestingly, we can obtain $k$-uniform mixed states with lower purity from the ones with a higher purity. For instance, from the 3-uniform mixed state with the purity $\frac{1}{2}$ in Example \ref{273} given by two pure states $|\Phi_{1}\rangle$ and $|\Phi_{2}\rangle$, we obtain a 3-uniform mixed state with purity $\frac{1}{4}$ given by four pure states $|\Phi_{1}\rangle$, $|\Phi_{2}\rangle$, $|\Phi_{1}'\rangle=(\mathbb{I}_{2^6}\otimes \sigma_x)|\Phi_{1}\rangle$ and
$|\Phi_{2}'\rangle=(\mathbb{I}_{2^6}\otimes \sigma_x)|\Phi_{2}\rangle.$

In addition, for any $k$-uniform state with respect to density matrix $\rho$, the purity of the $\rho$, ${\rm Tr}(\rho^2)$, satisfies $0<{\rm Tr}(\rho^2)\leq1$. When ${\rm Tr}(\rho^2)=1$, $\rho$ is a $k$-uniform pure state. Otherwise, it is a $k$-uniform mixed state \cite{shu}.

Concerning the range of the purity, it is observed that if the obtained state is pure, the corresponding purity is 1. Otherwise, the purity is less than or equal to $\frac{1}{2}$. The purity of each state derived from Theorem \ref{gsan} is $2^{m-N}$. Specifically, when $m=N$, the purity equals 1. Otherwise, it is less than or equal to $\frac{1}{2}$. The purity of every state from Theorem \ref{1un} is 1. In Theorem \ref{th}, the related purity is 1 for $m=1$, while the purity is equal to $\frac{1}{m}$ for $m\geq2$. The purity of all the $k$-uniform mixed states in Theorems \ref{cz}, \ref{dx} and Corollaries \ref{kjia1}, \ref{35} is less than or equal to $\frac{1}{d}$ for $d\geq2$.

With respect to the purity of the $k$-uniform states derived from Theorem \ref{th}, we denote $A_r$ an OA$(r,N,d,k)$ for given $N$, $d$ and $k$. Generally, the larger $r$ is, the smaller MD$(A_r)$ is. Therefore, the pivotal challenge is to determine the minimal number of runs $r$ and the minimal number of blocks $m$ if $A_r$ can be orthogonally partitioned into $m$ blocks. By minimizing both $r$ and $m$, we assert that the states obtained from Theorem \ref{th} attain the highest possible purity.

From Theorem \ref{th}, if an ${\rm OA}(r,N,d,k)$ can be orthogonally partitioned into $m$ blocks $\{A_1, A_2, \ldots, A_{m}\}$ and $\text{MD}(A_i) \geq k+1$, then we can obtain $k$-uniform states. When $r$ reaches the minimum for given $N$, $d$, and $k$, we use the following four cases (a), (b), (c) and (d) to illustrate that the obtained states in Section \ref{gw} attain the highest possible purity by our constructions. Nevertheless, for the case (e) we can not determine whether the purity attains the highest possible value.

(a) If $\frac{r}{m}=d^{k'}$ and ${\rm MD}(A_i)\geq k+1\geq N-k'+1$, then $A_i$ is an ${\rm OA}(d^{k'},N,d,k')$. In fact, assume that there exists a $d^{k'}\times k'$ subarray of $A_i$ containing two identical rows. Then ${\rm MD}(A_i)\leq N-k'$, which reaches a contradiction. However, if an ${\rm OA}(d^{k'},N,d,k')$ does not exist, then the ${\rm OA}(r,N,d,k)$ can not be orthogonally partitioned into $m$ blocks such that $\text{MD}(A_i) \geq k+1$.

In Example \ref{34}, there exits an orthogonal partition $\{A_1, A_2, \ldots, A_{9}\}$ of an ${\rm OA}(3^5,8,3,4)$ where $r=3^5$ reaches the minimum for given $N=8$, $d=3$, and $k=4$ from Table 12.2 in \cite{hss} and ${\rm MD}(A_i)\geq5$. Therefore, we assert that the 8-qutrit 4-uniform mixed state attain the highest possible purity $\frac{1}{9}$. This is because when $m=3$ we have $k'=4$. Since an ${\rm OA}(3^4,8,3,4)$ does not exist from Ref. \cite{hss}, any orthogonal partition with $m=3$ blocks of an ${\rm OA}(3^5,8,3,4)$ does not exist.

(b) If ${\rm MD}(A_i)\geq k+1=N$, then $\frac{r}{m}\leq d$. Therefore, $m=\frac{r}{d}$ reaches the minimum value. For example, the states obtained from Corollary \ref{kjia1} (1)-(3) attain the highest possible purity.

%For any $u,v=1,2,\ldots,\frac{r}{m}$, the non-negative integer $\Delta^{(uv)}$ is denoted as the number of coincidences between the $u$-th and $v$-th rows of $A_i$ \cite{muk95}. Let $\Delta=\sum_{1\leq u<v\leq \frac{r}{m}}\Delta^{(uv)}$.

(c) If ${\rm MD}(A_i)\geq k+1\geq N-k'$ for $k'\geq1$, then the number of coincidences between the $u$-th and $v$-th rows of $A_i$ satisfies that $\Delta^{(uv)}\leq k'$ \cite{muk95} and $\Delta_{\frac{r}{m}}=\sum_{1\leq u<v\leq \frac{r}{m}}\Delta^{(uv)}\leq k'C_{\frac{r}{m}}^2$. Furthermore, assume $\frac{r}{m}\geq d+1$, then for any $d+1$ rows, we have $\Delta_{d+1}=\sum_{1\leq u<v\leq d+1}\Delta^{(uv)}\leq k'C_{d+1}^2$ and $\Delta_{d+1}\geq N$. Therefore, if $N>k'C_{d+1}^2$, then $\frac{r}{m}\leq d$ and $m=\frac{r}{d}$ reaches the minimum value. Suppose $x_0$, $x_1$, $\ldots$, and $x_{d-1}$ are the number of occurrences of all levels in a column of $A_i$ such that $C_{x_0}^2+C_{x_1}^2+\cdots+C_{x_{d-1}}^2$ reaches the minimum value. Then $x_0+x_1+\cdots+x_{d-1}=\frac{r}{m}$ and $\Delta_{\frac{r}{m}}\geq N(C_{x_0}^2+C_{x_1}^2+\cdots+C_{x_{d-1}}^2)=N\times \frac{x_0^2+x_1^2+\cdots+x_{d-1}^2-\frac{r}{m}}{2}\geq N\times\frac{\frac{(\frac{r}{m})^2}{d}-\frac{r}{m}}{2}=\frac{rN(\frac{r}{m}-d)}{2md}$. Therefore, $k'C_\frac{r}{m}^2\geq \frac{rN(\frac{r}{m}-d)}{2md}$, otherwise $A_i$ does not exist and the ${\rm OA}(r,N,d,k)$ can not be orthogonally partitioned into $m$ blocks such that $\text{MD}(A_i) \geq k+1$.

For example, the minimal number of $r$ in an {\rm OA}$(r,7,3,5)$ is $3^6$ from Table 12.2 in \cite{hss}. If ${\rm MD}(A_i)\geq 6$, then $\frac{r}{m}\leq 3$ since $k'=1$ and  $7>C_4^2=6$. Therefore, the rows of an {\rm OA}$(3^6,7,3,5)$ can be orthogonal partitioned into at least $3^5$ blocks. Hence, the 5-uniform $7$-qutrit mixed state from the mixture of $3^5$ pure states in Example \ref{351} attains the highest possible purity.

(d) For the special case where $d=k=3$ and $N=5$, a 3-uniform mixed state of 5-qutrit with a purity of $\frac{1}{18}$ can be derived from Example \ref{18}. The smallest value of $r$ rows that permits the existence of an OA is 54. According to \cite{h97}, there are a total of four non-isomorphic OA$(54, 5,3,3)$s. However, none of these arrays can be orthogonally partitioned into nine blocks such that the minimum Hamming distance of each block is greater than or equal to 4. Therefore, $m=18$ is determined as the minimal value, and the resulting state achieves the highest possible purity.

(e) Whether $r$ reaches the minimum in any {\rm OA}$(r,N,d,k)$ for given $N$, $d$, and $k$ is uncertain. This uncertainty arises from the description provided in Table 12.1-12.3 of \cite{hss}, which indicates that the existence of an array with fewer than $r$ runs remains an open question.

We summarize our main results in Tables \ref{3sp} and \ref{4sp}, where ``P'' stands for the purity of the state, ``M'' the corresponding method, and ``H'' the items (a)-(e) above. That (a), (b), (c) and (d) appear in ``H'' column represents the states obtained achieve the highest possible purity.

\begin{table*}
\caption{$k$-uniform $N$-qutrit (pure and mixed) states}
\label{3sp}
\begin{center}
{\begin{tabular}{|c|c|c|c|c|c|c|c|c|c|c|c|c|c|c|c|c|c|c|}\hline
\diagbox{$k$}{$N$} &\multicolumn{3}{c|}{4} & \multicolumn{3}{c|}{5} & \multicolumn{3}{c|}{6} & \multicolumn{3}{c|}{7} & \multicolumn{3}{c|}{8} & \multicolumn{3}{c|}{9}\\\hline
 & P & M&H& P & M&H & P & M&H & P & M&H & P & M&H& P & M&H\\\hline
1 &1 &Th. \ref{1un} & & 1 &Th. \ref{1un}&& 1 &Th. \ref{1un} && 1 &Th. \ref{1un} && 1 &Th. \ref{1un} && 1 &Th. \ref{1un}&\\\hline
2&1 &Ref. \cite{npj} && 1 &Ref. \cite{npj}&& 1 &Ref. \cite{npj} && 1 &Ref. \cite{npj} && 1 &Ref. \cite{npj}& & 1 &Ref. \cite{npj}&\\\hline
3&$\frac{1}{9}$ &Col. \ref{kjia1}&(b)& $\frac{1}{18}$ &Ex. \ref{18}&(d)& 1 &Ref. \cite{npj} && 1 &Ref. \cite{npj} && 1 &Ref. \cite{npj}& &1 &Ref. \cite{npj}&\\\hline
4&&&& $\frac{1}{27}$ &Col. \ref{kjia1}&(b)&$\frac{1}{81}$&Ex. \ref{34}&(c)& $\frac{1}{27}$ &Ex. \ref{34} &(a)& $\frac{1}{9}$ &Ex. \ref{34}&(a)& 1 &Ref. \cite{huber18}&\\\hline
5&& &&& &&$\frac{1}{81}$ &Col. \ref{kjia1}&(b)& $\frac{1}{3^5}$&Ex. \ref{351}&(c)& $\frac{1}{81}$&Ex. \ref{351}&(a)& $\frac{1}{27}$ &Ex. \ref{351}&(a)\\\hline
6&&&&&&&&&&$\frac{1}{3^5}$ &Col. \ref{kjia1}&(b)&$\frac{1}{3^6}$&Col. \ref{kjia1}&(c)&$\frac{1}{3^7}$&Col. \ref{kjia1}& (e)\\\hline
7&&&&&&&&&&&&&$\frac{1}{3^6}$&Col. \ref{kjia1}&(b)& $\frac{1}{3^7}$&Col. \ref{kjia1}&(c) \\\hline
8&&&&&&&&&&&&&&&& $\frac{1}{3^7}$ &Col. \ref{kjia1}&(b)\\\hline
\end{tabular}}\\
\end{center}
\end{table*}

\begin{table*}
\caption{$k$-uniform $N$-ququart (pure and mixed) states}
\label{4sp}
\begin{center}
{\begin{tabular}{|c|c|c|c|c|c|c|c|c|c|c|c|c|c|c|c|c|c|c|}\hline
\diagbox{$k$}{$N$} &\multicolumn{3}{c|}{4} & \multicolumn{3}{c|}{5} & \multicolumn{3}{c|}{6} & \multicolumn{3}{c|}{7} & \multicolumn{3}{c|}{8} & \multicolumn{3}{c|}{9}\\\hline
 & P & M&H& P & M&H & P & M &H& P & M &H& P & M &H& P & M&H\\\hline
1 &1 &Th. \ref{1un}&& 1 &Th. \ref{1un}& & 1 &Th. \ref{1un}& & 1 &Th. \ref{1un}& & 1 &Th. \ref{1un}& & 1 &Th. \ref{1un}&\\\hline
2&1 &Ref. \cite{npj}& & 1 &Ref. \cite{npj}&& 1 &Ref. \cite{npj}& & 1 &Ref. \cite{npj}& & 1 &Ref. \cite{npj}& & 1 &Ref. \cite{npj}&\\\hline
3&$\frac{1}{16}$ &Col. \ref{kjia1}&(b)& $\frac{1}{4}$ &Ex. \ref{64}& (b) & 1 &Ref. \cite{npj}& & 1 &Ref. \cite{npj}& & 1 &Ref. \cite{npj}&& 1 &Ref. \cite{npj}&\\\hline
4&&& &$\frac{1}{64}$ &Col. \ref{kjia1}&(b)&$\frac{1}{4^4}$ &Ex. \ref{44}& (e) & $\frac{1}{64}$ &Ex. \ref{44}& (e) & $\frac{1}{16}$ &Ex. \ref{44}&(e) & 1 &Ref. \cite{huber18}&\\\hline
5&&&&&&& $\frac{1}{4^4}$ &Col. \ref{kjia1}&(b)& $\frac{1}{4^5}$&Ex. \ref{45}&(e)& $\frac{1}{4^4}$&Ex. \ref{45}&(e)& $\frac{1}{4^3}$ &Ex. \ref{45}&(e)\\\hline
6&&&&&&&&&&$\frac{1}{4^5}$ &Col. \ref{kjia1}&(b)&$\frac{1}{4^6}$&Col. \ref{kjia1}&(e)&$\frac{1}{4^5}$&Ex. \ref{46}&(e)\\\hline
7&&&&&&&&&&&&&$\frac{1}{4^6}$&Col. \ref{kjia1}&(b)& $\frac{1}{4^7}$& Col. \ref{kjia1}&(e)\\\hline
8&&&&&&&&&&&&&&&& $\frac{1}{4^7}$ &Col. \ref{kjia1}&(b)\\\hline
\end{tabular}}\\
\end{center}
\end{table*}

\section{Conclusions}\label{con}

We have introduced novel approaches in constructing $k$-uniform mixed states with the highest possible purity. These approaches have led to the discovery of various new $k$-uniform mixed states.

The techniques for the $N$-qubit case in Section \ref{zh} and those for the $N$-qudit case in Section \ref{gw} differ significantly. To date, the method described in Section \ref{zh} cannot be applied to the latter case, whereas the approach in Section \ref{gw} is applicable to the former. Notably, certain states derived from Theorem \ref{gsan} cannot be produced from Theorem \ref{th}, such as the 4-uniform 7-qubit mixed state with a purity $\frac{1}{8}$ in Example \ref{74}, due to the unequal coefficients in terms of $|\varphi_1\rangle$. All methods introduced in Section \ref{gw} are uniform.

OAs have been playing a crucial role in constructing $k$-uniform states. We have constructed $k$-uniform pure states by using IrOAs \cite{npj,qinp,jpa}. In this paper, OAs with orthogonal partitions are employed for constructing $k$-uniform mixed states (see Section \ref{gw}). Additionally, the Galois Field GF(4) is utilized for constructing $k$-uniform mixed qubit states (see Section \ref{zh}). See FIG. \ref{A} for details.

Significant differences between our findings and those reported in Ref. \cite{kmix} include the following:

1. Our approach in Section \ref{zh} utilizes a matrix over Galois field GF(4), whereas the method in \cite{kmix} involves identifying a set of $N$-qubit Pauli operators, which is more complex.

2. We propose a general method for generating $k$-uniform mixed states by using OAs in Section \ref{gw}, which is a strategy not offered in Ref. \cite{kmix}. Our methods are uniform and capable of producing mixed states in any dimension, which is a greater challenge than constructing solely two-dimensional states.

3. We have constructed all the states not depicted in Fig. 1 of Ref. \cite{kmix}.

4. The coefficients of each term in the pure states of the mixed states obtained through our methods in Section \ref{gw} are all equal to 1 under normalization, which is different from those in Ref. \cite{kmix}.

5. Moreover, many states produced by our methods exhibit higher purity compared with those in Ref. \cite{kmix}.

\thispagestyle{empty}
% 流程图定义基本形状
\tikzstyle{jx}=[thick, rounded corners, rectangle, minimum width=2cm, minimum height=1cm, text centered, draw=black] % 矩形
\tikzstyle{lx}=[diamond, aspect=2.3, text centered, draw=black] %菱形
\tikzstyle{arrow}=[->,>=stealth]% 箭头形式

\begin{figure*}[ht]
\centering
\begin{tikzpicture}
%定义流程图具体形状
\node[jx](start){$k$-uniform state};
\node[jx, left of=start, yshift=-1.5cm, node distance=4.5cm](p){$k$-uniform pure state};
\node[jx, right of=start, yshift=-1.5cm, node distance=4.5cm](m){$k$-uniform mixed state};
\node[jx, left of=p, yshift=-1.5cm,node distance=2.5cm](pb){qubit};
\node[jx, right of=p, yshift=-1.5cm,node distance=2.5cm](pd){qudit};
\node[jx, left of=m, yshift=-1.5cm,node distance=2.5cm](mb){qubit};
\node[jx, right of=m, yshift=-1.5cm,node distance=2.5cm](md){qudit};
\node[jx, below of=p, yshift=-1cm,node distance=2cm,draw=blue](hh){IrOAs};
\node[jx, left of=mb, yshift=-1.5cm,node distance=2cm,draw=blue](gf){GF(4)};
\node[jx, below of=m, yshift=-1cm,node distance=2cm,draw=blue](fh){orthogonal partitions of OAs};
\draw [->, thick] (start) -| (p);
\draw [->, thick] (start) -| (m);
\draw [->, thick] (p) -| (pb);
\draw [->, thick] (p) -| (pd);
\draw [->, thick] (m) -| (mb);
\draw [->, thick] (m) -| (md);
\draw [->, thick] (pb) |- (hh);
\draw [->, thick] (pd) |- (hh);
\draw [->, thick] (mb) -| (gf);
\draw [->, thick] (mb) -| (fh);
\draw [->, thick] (md) -| (fh);
\end{tikzpicture}
\caption{Construction methods for $k$-uniform states. Black lines represent the methods.}
\label{A}
\end{figure*}
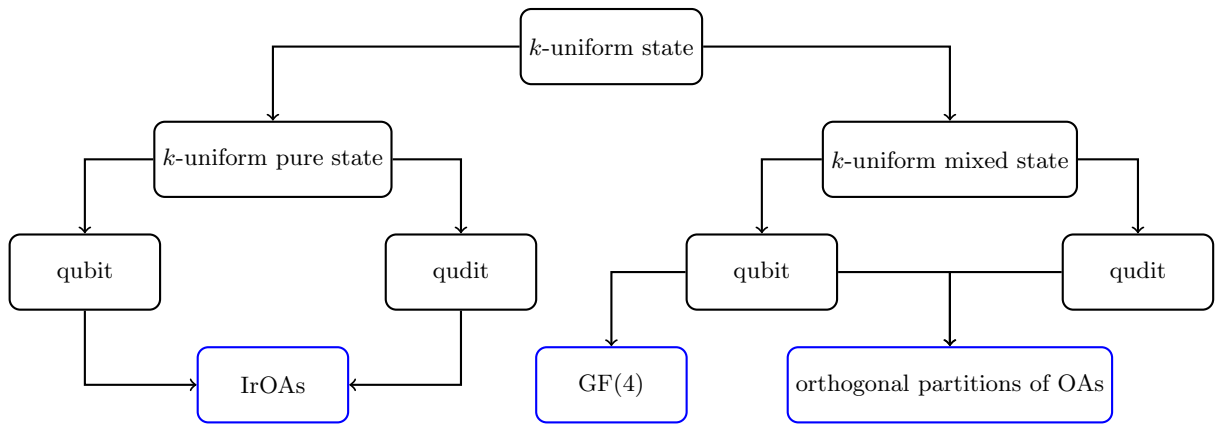

\section*{ACKNOWLEDGEMENTS}

{\bf Funding:} This work was supported by the National Natural Science Foundation of China (Grant Nos. 11971004, 12075159 and 12171044); the
specific research fund of the Innovation Platform for Academicians of Hainan Province under Grant No. YSPTZX202215; Key Scientific and Technological Research Projects in Henan Province(Grant No. 242102220105).

{\bf Conflicts of Interest:} The authors declare no conflict of interest.

{\bf Data Availability Statement:} The data that support the findings of this study are available from
the corresponding author upon reasonable request.

\end{document}